\documentclass[pra,twocolumn,floatfix,superscriptaddress,longbibliography,notitlepage]{revtex4-1}

\usepackage{graphicx, color, graphpap}
\usepackage{enumitem}
\usepackage{amssymb}
\usepackage{amsthm}
\usepackage{multirow}
\usepackage[colorlinks=true,citecolor=blue,linkcolor=magenta]{hyperref}
\usepackage[T1]{fontenc}
\usepackage{bbm}
\usepackage{thmtools,thm-restate}
\usepackage{verbatim}
\usepackage{mathtools}
\usepackage{titlesec}
\usepackage{amsmath}
\usepackage[title]{appendix}

\usepackage[caption = false]{subfig}

\long\def\ca#1\cb{} 

\newcommand{\poly}{\operatorname{poly}}

\newcommand{\EC}{\mathcal{E}}

\newcommand{\OC}{\mathcal{O}}

\newcommand{\Tr}{{\rm Tr}}

\newcommand{\Var}{{\rm Var}}

\renewcommand{\geq}{\geqslant}
\renewcommand{\leq}{\leqslant}

\renewcommand{\vec}[1]{\boldsymbol{#1}}  

\newcommand{\ad}{^\dagger}

\newcommand{\thv}{\vec{\theta}}

{}
{}
\newtheorem{corollary}{Corollary}
\newtheorem{proposition}{Proposition}

\begin{document}

\title{Higher Order Derivatives of Quantum Neural Networks with Barren Plateaus}

\author{M. Cerezo}
\affiliation{Theoretical Division, Los Alamos National Laboratory, Los Alamos, NM 87545, USA}
\affiliation{Center for Nonlinear Studies, Los Alamos National Laboratory, Los Alamos, NM, USA
}

\author{Patrick J. Coles}
\affiliation{Theoretical Division, Los Alamos National Laboratory, Los Alamos, NM 87545, USA}

\begin{abstract}
Quantum neural networks (QNNs) offer a powerful paradigm for programming near-term quantum computers and have the potential to speedup applications ranging from data science to chemistry to materials science.
However, a possible obstacle to realizing that speedup is the Barren Plateau (BP) phenomenon, whereby the gradient vanishes exponentially in the system size $n$ for certain QNN architectures. The question of whether high-order derivative information such as the Hessian could help escape a BP was recently posed in the literature. Here we show that the elements of the Hessian are exponentially suppressed in a BP, so estimating the Hessian in this situation would require a precision that scales exponentially with $n$. Hence, Hessian-based approaches do not circumvent the exponential scaling associated with BPs. We also show the exponential suppression of higher order derivatives. Hence, BPs will impact optimization strategies that go beyond (first-order) gradient descent. In deriving our results, we prove novel, general formulas that can be used to analytically evaluate any high-order partial derivative on quantum hardware. These formulas will likely have independent interest and use for training quantum neural networks (outside of the context of BPs). 
\end{abstract}

\maketitle

\section{Introduction}

Standard quantum algorithms were not developed to handle the constraints imposed by current quantum computers. Such Noisy Intermediate Scale Quantum (NISQ) devices have limited connectivity, limited qubit count, and noise that limits circuit depth.

On the other hand, training parameterized quantum circuits provides a promising approach for quantum computing in the NISQ era, as this approach adapts to the imposed constraints. Here, one utilizes a quantum computer to efficiently evaluate a cost (or loss) function $C(\thv)$ or its gradient $\nabla C(\thv)$, while employing a classical optimizer to train the parameters $\thv$ of a parameterized quantum circuit $V(\thv)$. This strategy is employed in two closely related paradigms: Variational Quantum Algorithms (VQAs)~\cite{cerezo2020VQAreview} for chemistry, optimization, and other applications~\cite{VQE,mcclean2016theory,qaoa2014,Romero,QAQC,VQSD,arrasmith2019variational,cerezo2020variationalfidelity,cirstoiu2019variational,bravo-prieto2019,xu2019variational,cerezo2020variational,kyriienko2020solving,chivilikhin2020mog,sharma2019noise}, and Quantum Neural Networks (QNNs) for classification applications~\cite{schuld2014quest,cong2019quantum,beer2020training,verdon2018universal,abbas2020power}. QNNs can be viewed as a generalization of VQAs to the case of multiple input states, and thus we will henceforth use the term QNNs to encompass all methods that train parameterized quantum circuits.

While many novel QNNs have been developed, more rigorous scaling analysis is needed for these architectures. One of the few known results is the so-called barren plateau phenomenon~\cite{mcclean2018barren,cerezo2020cost,sharma2020trainability,wang2020noise,holmes2020barren,marrero2020entanglement,uvarov2020barren,arrasmith2020effect,holmes2021connecting,patti2020entanglement}, where the cost function gradient vanishes exponentially with the system size. This can arise due to deep unstructured ansatzes~\cite{mcclean2018barren,sharma2020trainability,holmes2021connecting}, global cost functions~\cite{cerezo2020cost,sharma2020trainability}, noise~\cite{wang2020noise}, or an excess of entanglement~\cite{marrero2020entanglement,patti2020entanglement}. Regardless of the origin, when a cost landscape exhibits a barren plateau, one requires an exponential precision to determine a minimizing direction in order to navigate the landscape. Since the standard goal of quantum algorithms is polynomial scaling with the system size (in contrast to the exponential scaling of classical algorithms), the exponential scaling due to barren plateaus can destroy quantum speedup. Hence, the study and analysis of barren plateaus should be viewed as a fundamental step in the development of QNNs to guarantee that they can, in fact, provide a speedup over classical algorithms.

Recently there have been multiple strategies proposed for avoiding barren plateaus such as  employing local cost functions~\cite{cerezo2020cost,uvarov2020barren},  pre-training~\cite{verdon2019learning}, parameter correlation~\cite{volkoff2021large}, layer-by-layer training~\cite{skolik2020layerwise}, initializing layers to the identity~\cite{grant2019initialization}, and employing problem-inspired ansatzes~\cite{bharti2020iterative,bharti2020quantumAS}. These strategies are aimed at either avoiding or preventing the existence of a barren plateau, and they appear to be promising, with more research needed on their efficacy on general classes of problems. In a recent article~\cite{Huembeli2020Characterizing}, an alternative idea was proposed involving a method for actually training inside and escaping a barren plateau. Specifically, the proposal was to compute the Hessian $H$ of the cost function, and the claim was that taking a learning rate proportional to the inverse of the largest eigenvalue of the Hessian leads to an optimization method that could escape the barren plateau.


The question of whether higher-order derivative information (beyond the first-order gradient) is useful for escaping a barren plateau is interesting and is the subject of our work here. Our main results are presented here in the form of two propositions and corollaries. First, we show that the matrix elements $H_{ij}$ of the Hessian are exponentially vanishing when the cost exhibits a barren plateau. This implies that the calculation of  $H_{ij}$  requires exponential precision. In our second result we show that the magnitude of any higher-order partial derivative of the cost will also be exponentially small in a barren plateau. Our results suggest that optimization methods that use higher-order derivative information, such as the Hessian, will also face exponential scaling, and hence do not circumvent the scaling issues arising from barren plateaus.

As a byproduct of our work, we derive novel formulas for higher-order partial derivatives of the cost function, which can be used to efficiently evaluate these derivatives on quantum hardware. These formulas are obtained via the so-called Pascal tree, which we also introduce here. Due to their generality, these formulas can be generically used for training parameterized quantum circuits. Higher-order derivative information can be useful for various applications such as chemistry~\cite{o2019calculating} and solving partial differential equations~\cite{kyriienko2020solving}, as well as characterizing landscapes~\cite{Huembeli2020Characterizing} and improving optimization methods (e.g., Newton's method~\cite{gill2019practical,mari2020estimating}).

\section{Preliminaries}

To set the stage for our results, we first give some background on the cost function, the parameter shift rule, and barren plateaus.

\subsection{Cost function}

In what follows, we consider the case when the cost can be expressed as a sum of expectation values: 
\begin{equation}\label{eq:cost}
    C(\thv)=\sum_{x=1}^N C_x,\quad \!\!\text{with}\quad \!\!C_x=\Tr[O_xV(\thv)\rho_x V\ad(\thv)]\,,
\end{equation}
where $\{\rho_x\}$ is a set (of size $N$) of input states to the parameterized circuit $V(\thv)$. In order for this cost to be efficiently computable, the number of states in the input set should grow at most polynomially with the number of qubits $n$, that is, $N\in\OC(\poly(n))$. In the context of QNNs, the states $\{\rho_x\}$ can be viewed as training data points, and hence \eqref{eq:cost} is a natural cost function for QNNs. In the context of VQAs, one typically chooses $N=1$, corresponding to a single input state. In this sense, the cost function in \eqref{eq:cost} is general enough to be relevant to both QNNs and VQAs.

\subsection{Parameter shift rule}

Let $\theta_i$ be an angle that parameterizes a unitary in $V(\thv)$ as $e^{-i \theta_i \sigma_i/2 }$,  with $\sigma_i$ a Hermitian operator with eigenvalues $\pm 1$. Then, the partial derivative $\frac{\partial C(\thv)}{\partial \theta_i}=\partial_i C(\thv)$ can be computed via the parameter shift rule~\cite{mitarai2018quantum,schuld2019evaluating} as
\begin{equation}\label{eq:pshift}
    \partial_i C(\thv)= \frac{1}{2}\left(C\left(\thv_{\overline{i}}, \theta_i^{(\frac{1}{2})}\right)-C\left(\thv_{\overline{i}}, \theta_i^{(-\frac{1}{2})}\right)\right)\,.
\end{equation}
Here, $\overline{i}$ denotes the indices distinct from $i$ (i.e., the indices of parameters that are not being differentiated), and we define the notation
\begin{equation}\label{eq:deltatheta}
    \theta_i^{ (\beta)}=\theta_i+\beta\pi.
    \end{equation}
Note that the parameter shift rule in~\eqref{eq:pshift} allows one to exactly write the first-order partial derivative as a difference of cost function values evaluated at two different points. Moreover, we remark that Eq.~\eqref{eq:pshift} is not a finite difference formula, but rather corresponds to exactly evaluating the partial derivative. 

\subsection{Barren plateaus}

As discussed in~\cite{mcclean2018barren,cerezo2020cost,sharma2020trainability,wang2020noise,holmes2020barren,marrero2020entanglement,uvarov2020barren,arrasmith2020effect,holmes2021connecting,patti2020entanglement}, by analyzing the scaling of the variance of the cost function partial derivative one can detect the presence of barren plateaus in the cost function landscape. We denote this variance as $\Var_{\thv}[\partial_i C]$, where the expectation values in the variance are taken over $\thv$. Specifically, when the cost function exhibits a barren plateau one finds that $\Var_{\thv}[\partial_i C]$ is exponentially vanishing with the number of qubits, i.e.,
\begin{equation}\label{eq:bpvar}
    \Var_{\thv}[\partial_i C]\leq F(n)\,, \quad \text{with} \quad F(n)\in \OC(1/b^n)\,,
\end{equation}
for some $b>1$.  Then, combining Eq.~\eqref{eq:bpvar} with Chebyshev’s inequality,  the probability that the cost  derivative deviates from its mean value (of zero~\cite{holmes2021connecting}) is bounded as
\begin{equation}\label{eq:CHineq0}
    \Pr \left(|\partial_i C|\geq c \right) \leq \frac{\Var_{\thv}[\partial_i C]}{c^2}\leq\frac{F(n)}{c^2}\,,
\end{equation}
for any $c>0$. Equation~\eqref{eq:CHineq0} shows that, on average, the cost function partial derivatives will be exponentially small across the landscape, meaning that an exponential precision is needed to estimate the gradients and determine a cost minimizing direction. 

Barren plateaus can arise due to multiple reasons. For instance, in the seminal work of~\cite{mcclean2018barren} it was shown that deep unstructured circuits (i.e., random parametrized circuits with depths in $\OC(\poly(n))$)  form $2$-designs on $n$ qubits and hence will have barren plateaus. This phenomenon was later extended in~\cite{cerezo2020cost} to the case of layered hardware efficient ansatzes where random two-qubit gates act on alternating pairs of qubits in a brick-like fashion. Therein it was shown that the existence of barren plateaus is related to the locality of the operators  $O_x$ in~\eqref{eq:cost}. Global operators $O_x$ that act non-trivially on all qubits have barren plateaus irrespective of the depth of $V(\thv)$, while local operators which measure individual qubits will have barren plateaus only for deep circuits. Barren plateaus were then identified in perceptron-based quantum neural networks~\cite{sharma2020trainability}, and  in tasks of learning scramblers~\cite{holmes2020barren}. Moreover, the barren plateau phenomenon has also been linked to the presence of high levels of entanglement in the circuit~\cite{marrero2020entanglement,patti2020entanglement} and to the hardware noise acting throughout the computation~\cite{wang2020noise}.

In what follows, we formulate our main results irrespective of the mechanism that leads to the barren plateau, and hence our results apply to all such mechanisms.

\section{Hessian matrix elements}

Let us now state our results of how the presence of barren plateaus can affect the estimation of the Hessian. The Hessian $H$ of the cost function is a square matrix whose matrix elements are the second derivatives of $C(\thv)$, i.e.,
\begin{equation}
    H_{ij}=\frac{\partial^2 C(\thv)}{\partial\theta_i\partial\theta_j}=\partial_i\partial_j C(\thv)\,.
\end{equation}
Reference~\cite{Huembeli2020Characterizing} noted that the matrix elements of the Hessian can be written according to the parameter shift rule. Namely, one can first write
\begin{equation}\label{eq:pdj0}
    H_{ij}=\frac{1}{2}\left[ \partial_i C\left(\thv_{\overline{j}}, \theta_j^{(\frac{1}{2})}\right)-\partial_i C\left(\thv_{\overline{j}}, \theta_j^{(-\frac{1}{2})}\right)\right]
\end{equation}
and then apply the parameter shift rule a second time:
\small
\begin{align}
    H_{ij}=\frac{1}{4}\Big[&C \left(\thv_{\overline{ij}}, \theta_i^{(\frac{1}{2})},\theta_j^{(\frac{1}{2})}\right)+
    C \left(\thv_{\overline{ij}}, \theta_i^{(-\frac{1}{2})},\theta_j^{(-\frac{1}{2})}\right)\label{eq:pdj}\\
    -&C \left(\thv_{\overline{ij}}, \theta_i^{(\frac{1}{2})},\theta_j^{(-\frac{1}{2})}\right)-
    C \left(\thv_{\overline{ij}}, \theta_i^{(-\frac{1}{2})},\theta_j^{(\frac{1}{2})}\right)
    \Big]\nonumber\,.
\end{align}
\normalsize
Now, the second derivatives of the cost can be expressed as a sum of cost functions being evaluated at (up to) four points.  

From the parameter shift rule we can then derive the following bound on the probability that the  magnitude of the matrix elements $|H_{ij}|$ are larger than a given $c>0$.

\begin{proposition}\label{prop1}
Consider a cost function of the form~\eqref{eq:cost}, for which the parameter shift rule of~\eqref{eq:pshift} holds. 
Let $H_{ij}$ be the matrix elements of the Hessian as defined in~\eqref{eq:pdj0}. Then, assuming that $\langle\partial_i C\rangle_{\thv}=0$, the following inequality holds  for any $c>0$,
\begin{align}\label{eq:unionprop}
   \Pr(|H_{ij}|\geq c)
   &\leq  \frac{2\Var_{\thv}[\partial_i C]}{c^2}\,.
\end{align}
Here  $\Var_{\thv}[\partial_i C]=\langle(\partial_i C)^2\rangle_{\thv}-\langle \partial_i C\rangle^2_{\thv}$\,, where the expectation values are taken over $\thv$. 
\end{proposition}

\begin{proof}
Equation~\eqref{eq:pdj0} implies that the magnitudes of the Hessian matrix elements are bounded as
\begin{align}
    |H_{ij}|\leq  \frac{1}{2}\left( \left|\partial_i C\left(\thv_{\overline{j}}, \theta_j^{(\frac{1}{2})}\right)\right|+\left|\partial_i C\left(\thv_{\overline{j}}, \theta_j^{(-\frac{1}{2})}\right)\right|\right)\label{eq:upper}\,.
\end{align}
From Chebyshev’s inequality we can bound the probability that the cost  derivative deviates from its mean value (of zero) as
\begin{equation}\label{eq:CHineq}
    \Pr \left(|\partial_i C|\geq c \right) \leq \frac{\Var_{\thv}[\partial_i C]}{c^2}\,,
\end{equation}
for all $c>0$, and for all $i$.  Then, let $\EC_\pm$ be defined as the event that $\left|\partial_i C(\thv_\pm)\right| \geq c$, where $\thv_\pm=\left(\thv_{\overline{j}}, \theta_j^{(\pm\frac{1}{2})}\right)$. Note that the set of events where $|H_{ij}|\geq c$ is a subset of the set $\EC_+\cup\EC_-$. Then, from the union bound and Eq.~\eqref{eq:CHineq} we can recover~\eqref{eq:unionprop} as follows:
\begin{align}\label{eq:union}
   \Pr(|H_{ij}|\geq c)
   &\leq  \Pr(\EC_+\cup\EC_-)\\
   &\leq \Pr(\EC_+)+ \Pr(\EC_-)\\
   &\leq  \frac{\Var_{\thv_+}[\partial_i C]}{c^2}+ \frac{\Var_{\thv_-}[\partial_i C]}{c^2}\\
   &=  \frac{2\Var_{\thv}[\partial_i C]}{c^2}\,,
\end{align}
where we used the fact that $\langle\cdot\rangle_{\thv}=\langle\cdot\rangle_{\thv_\pm}$.
\end{proof}

Then, the following corollary holds

\begin{corollary}\label{cor1}
Consider the bound in Eq.~\eqref{eq:union} of Proposition~\ref{prop1}. If the cost exhibits a barren plateau, such that~\eqref{eq:bpvar} holds,  then the matrix elements of the Hessian are exponentially vanishing since 
\begin{align}\label{eq:unionpropcor}
   \Pr(|H_{ij}|\geq c)
   &\leq  \frac{2F(n)}{c^2}\,,
\end{align}
where $F(n)\in\OC(1/b^n)$ for some $b>1$.
\end{corollary}

The proof follows by combining~\eqref{eq:unionprop} and~\eqref{eq:bpvar}. Corollary~\ref{cor1} shows that when the cost landscape exhibits a barren plateau, the matrix elements of the Hessian are exponentially vanishing with high probability. This implies that any algorithm that requires the estimation of the Hessian will requires a precision that grows exponentially with the system size.

\section{Higher order partial derivatives}

Let us now analyze the magnitude of  higher order partial derivatives in a barren plateau.  We use the following notation for the $|\vec{\alpha}|$th-order derivative
\begin{align}\label{eq:horderpd}
D^{\vec{\alpha}} C(\thv)&=\partial_{\alpha_1}\partial_{\alpha_2}\cdots \partial_{\alpha_{|\vec{\alpha}|}}C(\thv)\,,
\end{align}
where ${\vec{\alpha}}$ is an $|\vec{\alpha}|$-tuple.  Since one can take the derivative with respect to the same angle multiple times, we define the set $\vec{\Theta}$ (of size $M=|\vec{\Theta}|$) as the set of distinct angles with respect to which we take the partial derivative. Similarly, let $\overline{\vec{\Theta}}$ be the compliment of $\vec{\Theta}$, so that $\vec{\Theta}\cup\overline{\vec{\Theta}}=\vec{\theta}$. Then, for any $\Theta_k\in \vec{\Theta}$ we define $N_k$ as the multiplicity of $\Theta_k$ in $\vec{\alpha}$ such that $\sum_{k=1}^M N_k=|\vec{\alpha}|$. Since the cost function and any of its higher order partial derivatives are continuous function of  the parameters (as can be seen below via multiple applications of the parameter shift rule), one can extend Clairaut's Theorem~\cite{stewart2015multivariable}  to rewrite
\begin{align}\label{eq:higherorder}
D^{\vec{\alpha}} C(\thv)&=\partial^{N_1}_{\Theta_{1}}\cdots \partial^{N_M}_{\Theta_{M}}C(\thv)\,.
\end{align}

\begin{figure}[t]
    \centering
    \includegraphics[width=.9\columnwidth]{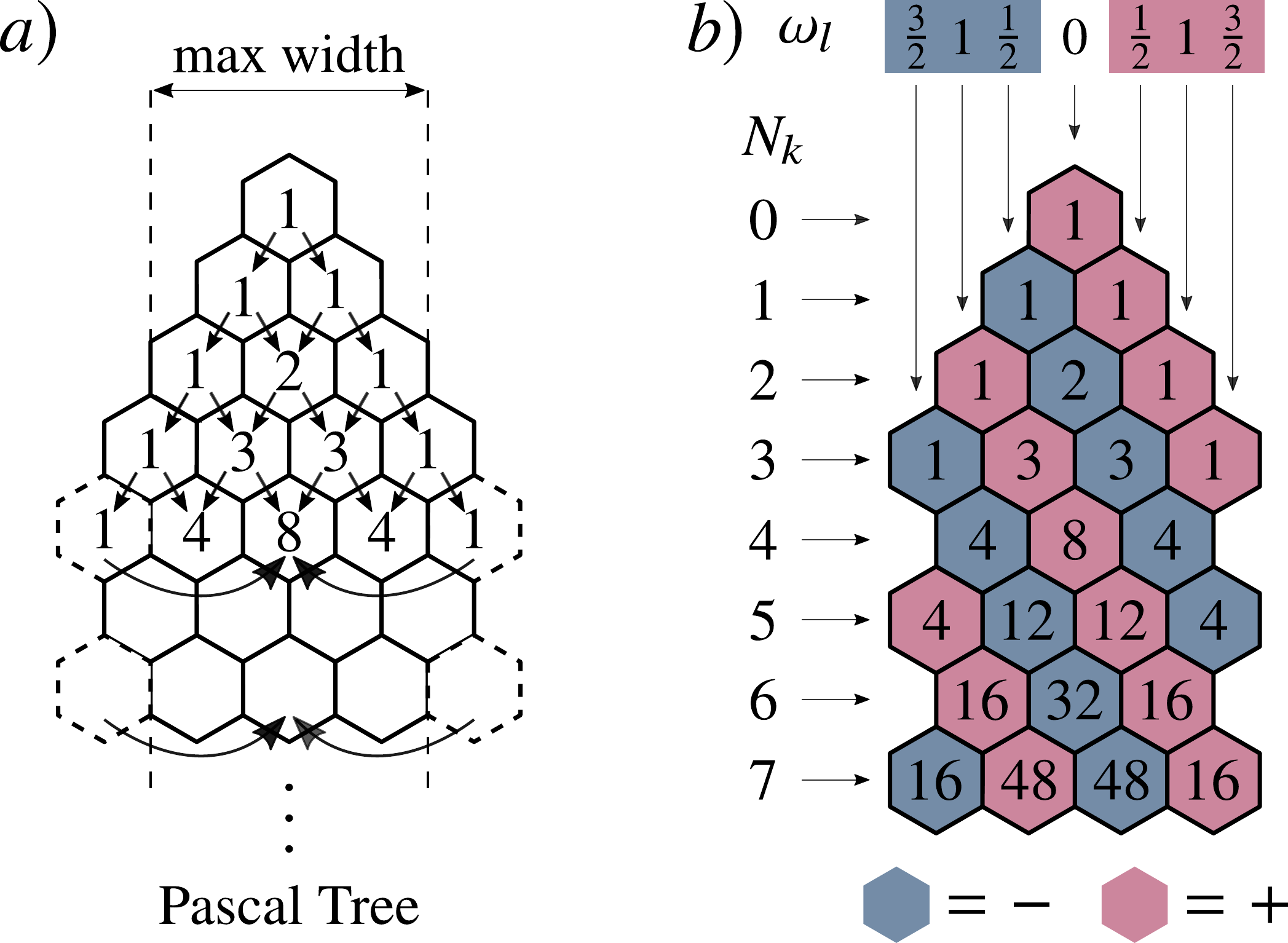}
    \caption{\textbf{The Pascal tree.} a) The Pascal tree can be obtained by modifying how a Pascal triangle is constructed. 
    In a Pascal triangle each entry of a row is obtained by adding together the numbers directly above to the left and above to the right, with blank entries considered to be equal to zero. The entries of a  Pascal tree  are obtained following the aforementioned rule, with the  additional constraint that the width of the triangle is restricted to always being smaller than a given even  number.  Moreover, once an entry in a row is outside the maximum width, its value is added to the central entry in that row (see arrows). Here the maximum width is four. b) The coefficients $d_{(\omega_l,N_l)}$ in~\eqref{eq:ds} can be obtained from the Pascal tree of (a) by adding signs to the entries of the tree. As schematically depicted, all entries in a diagonal going from top left to bottom right have the same sign, with the first entry in the first row having a positive sign. Here, each row corresponds to a given $N_l$, while entries in a row correspond to different values of $\omega_l$, with $\omega_l\in\{0,\pm 1\}$ if $N_l$ is even, and $\omega_l\in\{\pm\frac{1}{2},\pm\frac{3}{2}\}$ if $N_l$ is odd. For instance, $d_{(-\frac{1}{2},5)}=-12$.    }
    \label{fig:pascal}
\end{figure}

Then, applying the parameter shift rule  $|\vec{\alpha}|$ times we find that the $|\vec{\alpha}|$-order partial derivative can be expressed as a summation of cost functions evaluated at (up to) $2^{|\vec{\alpha}|}$ points as
\begin{equation}\label{eq:expansioncost}
 D^{\vec{\alpha}} C(\thv)=\frac{1}{2^{|\vec{\alpha}|}}\sum_{\vec{\omega}} W_{\vec{\omega}} C(\overline{\vec{\Theta}},\vec{\Theta}^{(\vec{\omega})})\,.
\end{equation}
Here we defined $\vec{\Theta}^{(\vec{\omega})}=(\Theta_1^{(\omega_1)},\ldots,\Theta_M^{(\omega_M)})$, with $\Theta_k^{(\omega_k)}=\Theta_k+\omega_k\pi$  defined analogously to~\eqref{eq:deltatheta},  and where
\begin{equation}
    W_{\vec{\omega}}=\prod_{l=1}^M d_{(\omega_l,N_l)} \quad \text{such that} \quad \sum_{\vec{\omega}}\left|W_{\vec{\omega}}\right|=2^{|\vec{\alpha}|}\,.\label{eq:ds}
\end{equation}
Also, $\vec{\omega}=(\omega_1,\ldots,\omega_{D})$, where 
$\omega_l\in\{0,\pm 1\}$ if $N_l$ is even, and $\omega_l\in\{\pm\frac{1}{2},\pm\frac{3}{2}\}$ if $N_l$ is odd.
Additionally, the coefficients $d_{\omega_l,N_k}$ can be obtained from the Pascal tree which we introduce in Fig.~\ref{fig:pascal}. In the Appendix we provide additional details regarding the coefficients $d_{(\omega_l,N_l)}$ and the  Pascal tree.

From~\eqref{eq:expansioncost} we obtain that the $(|\vec{\alpha}|+1)$th-order derivative, which we denote as $ \partial_iD^{\vec{\alpha}} C(\thv)= D^{i,\vec{\alpha}}C(\thv)$, is obtained as the sum of (up to) $2^{|\vec{\alpha}|}$ partial derivatives:
\begin{equation}\label{eq:alphaiderivative}
  D^{i,\vec{\alpha}} C(\thv)=\frac{1}{2^{|\vec{\alpha}|}}\sum_{\vec{\omega}} W_{\vec{\omega}} \partial_i C(\overline{\vec{\Theta}},\vec{\Theta}^{(\vec{\omega})})\,.
\end{equation}
Since one has to individually evaluate each term in~\eqref{eq:alphaiderivative} and since there are up to $2^{|\vec{\alpha}|}$ terms, we henceforth assume that  $|\vec{\alpha}|\in\OC(\log(n))$. This guarantees that the computation of $  D^{i,\vec{\alpha}} C(\thv)$ leads to an overhead which is (at most) $\OC(\poly(n))$.

The following proposition, which generalizes Proposition~\ref{prop1},  allows us to bound the probability that the magnitude of $\left|D^{i,\vec{\alpha}} C(\thv)\right |$ is larger than a given $c>0$.

\begin{proposition}\label{prop2}
Consider a cost function of the form~\eqref{eq:cost}, for which the parameter shift rule of~\eqref{eq:pshift} holds. Let $D^{i,\vec{\alpha}} C(\thv)$ be a higher order partial derivative of the cost as defined in~\eqref{eq:horderpd}.  Then, assuming that $\langle\partial_i C\rangle_{\thv}=0$, the following inequality holds for any $c>0$,
\begin{align}\label{eq:unionprop2}
   \Pr(\left|D^{i,\vec{\alpha}} C(\thv)\right |\geq c)
   &\leq  \frac{2^{|\vec{\alpha}|}\Var_{\thv}[\partial_i C]}{c^2}\,.
\end{align}
\end{proposition}

\begin{proof}

From Eq.~\eqref{eq:alphaiderivative} we can obtain the following bound  
\begin{equation}\label{eq:alphaiderivative2}
 \left|D^{i,\vec{\alpha}} C(\thv)\right |\leq\frac{1}{2^{|\vec{\alpha}|}}\sum_{\vec{\omega}} \left|W_{\vec{\omega}}\right| \left| \partial_i C(\overline{\vec{\Theta}},\vec{\Theta}^{(\vec{\omega})})\right|\,.
\end{equation}
Let us define $\EC_{\vec{\omega}}$ as the event that $\left| \partial_i C(\overline{\vec{\Theta}},\vec{\Theta}^{(\vec{\omega})})\right|\geq c$. Since~\eqref{eq:bpvar} holds, then the following chain of inequalities holds
\begin{align}
   \Pr(\left|D^{i,\vec{\alpha}} C(\thv)\right |\geq c)   &\leq  \Pr\left(\bigcup_{\vec{\omega}}\EC_{\vec{\omega}}\right)\\
   &\leq \sum_{\vec{\omega}} \Pr\left(\EC_{\vec{\omega}}\right)\\
   &\leq \sum_{\vec{\omega}} \frac{\Var_{(\overline{\vec{\Theta}},\vec{\Theta}^{(\vec{\omega})})}[\partial_i C]}{c^2}\label{eq:summation}\\
   &\leq  \frac{2^{|\vec{\alpha}|}\Var_{\thv}[\partial_i C]}{c^2}\,,\label{eq:union2}
\end{align}
where we invoked the union bound, and where we recall that $\langle\cdot\rangle_{\thv}=\langle\cdot\rangle_{(\overline{\vec{\Theta}},\vec{\Theta}^{(\vec{\omega})})}$, $\forall \vec{\omega}$. In addition, for~\eqref{eq:union2} we used the fact that the summation in~\eqref{eq:summation} has at most $2^{|\vec{\alpha}|}$ terms.
\end{proof}

Then, if the cost function exhibits a barren plateau, the following corollary follows. 
\begin{corollary}\label{cor2}
Consider the bound in Eq.~\eqref{eq:unionprop2} of Proposition~\ref{prop2}. If the cost exhibits a barren plateau, such that~\eqref{eq:bpvar} holds, then higher order partial derivatives of the cost function are exponentially vanishing since 
\begin{align}\label{eq:corolary2}
   \Pr(\left|D^{i,\vec{\alpha}} C(\thv)\right |\geq c)
   &\leq  \frac{G(n)}{c^2}\,, 
\end{align}
where $G(n)\in \OC(1/q^{n})$ for some $q>1$.
\end{corollary}

\begin{proof}
Combining~\eqref{eq:bpvar} and~\eqref{eq:unionprop2} leads to 
\begin{align}\label{eq:corolary22}
   \Pr(\left|D^{i,\vec{\alpha}} C(\thv)\right |\geq c)
   &\leq  \frac{2^{|\vec{\alpha}|}F(n)}{c^2}\,.
\end{align}
Then, let us define $G(n)=2^{|\vec{\alpha}|}F(n)$. Since $|\vec{\alpha}|\in\OC(\log(n))$, and  $F(n)\in\OC(1/b^n)$, then  we know that there exists  $\kappa$, $\kappa'$, and $n_0$ such that $\forall n>n_0$ we respectively have $2^{|\vec{\alpha}|}\leq n^\kappa$ and $F(n)\leq \frac{\kappa '}{b^n}$.  Combining these two results we find  
\begin{equation}\label{eq:Geq}
    G(n)\leq\frac{\kappa'n^\kappa}{b^n}=\frac{\kappa'}{b^{L(n)}}\,, \quad \forall n>n_0\,,
\end{equation}
where $L(n)=(n-\kappa\log_b(n))$. Equation~\eqref{eq:Geq} shows that $G(n)\in\OC(1/b^{L(n)})$. 
Then, since
\begin{equation}
    \lim_{n\rightarrow \infty}\frac{L(n)}{n}=1\,,
\end{equation}
we have $L(n)\in\Omega(n)$, meaning  that there exists a $\widehat{\kappa}>0$ and $\widehat{n}_0$ such that $\forall n>\widehat{n}_0$, we have $L(n)\geq \widehat{\kappa}n $. The latter implies  $G(n)\leq\frac{\kappa'}{b^{\widehat{\kappa}n}}$  for all $n> \max\{n_0,\widehat{n}_0\}$, which  means that  
$G(n)\in \OC(1/q^{n})$ where $q=b^{\widetilde{\kappa}}$. Also, $q>1$ follows from $b>1$ and $\widetilde{\kappa} > 0$.
\end{proof}

Corollary~\ref{cor2} shows that, in a barren plateau, the magnitude of any efficiently computable higher order partial derivative (i.e., any partial derivative where $|\vec{\alpha}|\in\OC(\log(n))$) is exponentially vanishing in $n$ with high probability.

\section{Discussion}

In this work, we investigated the impact of barren plateaus on higher order derivatives. As shown in~\cite{Huembeli2020Characterizing} and~\cite{abbas2020power}, information on these higher-order derivatives can be used to analyze the landscape of cost functions and shed some light on the relatively obscure nature of training landscapes for Variational Quantum Algorithms and Quantum Neural Networks. Moreover, it was also suggested that one could use higher order derivative information to escape a barren plateau. We considered a cost function $C$ that is relevant to both Variational Quantum Algorithms and Quantum Neural Networks, as barren plateaus are relevant to both of these applications.


Our main result was that, when a barren plateau exists, the Hessian and other high order partial derivatives of $C$ are exponentially vanishing in $n$ with high probability. Our proof relied on the parameter shift rule, which we showed can be applied iteratively to relate higher order partial derivatives to the first order partial derivative (analogous to what Ref.~\cite{Huembeli2020Characterizing} did for the Hessian). Hence, the parameter shift rule allowed us to state the vanishing of higher order derivatives as essentially a corollary of the vanishing of the first order derivative. We remark that iterative applications of the parameter shift rule led us to a mathematically interesting construct that we called the Pascal tree, depicted in Fig.~\ref{fig:pascal}.

Our results imply that estimating higher order partial derivatives in a barren plateau is exponentially hard. Hence, any optimization strategy that requires information about partial derivatives that go beyond first-order (such as the Hessian) will require a precision that grows exponentially with $n$. We therefore surmise that, by themselves, optimizers that go beyond first order gradient descent do not appear to be a feasible solution to the barren plateau problem. More generally, our results suggest that it is better to develop strategies that avoid the appearance of the barren plateau altogether, rather than to try to escape an existing barren plateau.   

Finally, we remark that our results were derived using Eq.~\eqref{eq:alphaiderivative} and the Pascal tree, which provide novel, general formulas for arbitrary higher-order partial derivatives of the cost function in Eq.~\eqref{eq:cost}. These formulas are of interest on their own, as they can be generically employed to analytically evaluate higher-order derivatives of the cost on quantum hardware. As an example application, one could apply this method to take derivatives of quantum feature maps for solving partial differential equations~\cite{kyriienko2020solving}.

While some prior work on higher-order derivative formulas was performed in Ref.~\cite{mitarai2020theory}, our formulas are more explicit and our introduction of the Pascal tree is novel. We also remark that, in a recent post that was concurrent with our work, Ref.~\cite{mari2020estimating} introduced a formula similar to Eq.~\eqref{eq:alphaiderivative} for explicitly determining higher-order partial derivatives, albeit that work analyzes the formulas in a different context.


\textit{Acknowledgements.---} We thank Kunal Sharma for helpful discussions. Research presented in this article was supported by the Laboratory Directed Research and Development program of Los Alamos National Laboratory under project number 20180628ECR.   MC was also supported by the Center for Nonlinear Studies at LANL. PJC also acknowledges support from the LANL ASC Beyond Moore's Law project. This work was also supported by the U.S. Department of Energy (DOE), Office of Science, Office of Advanced Scientific Computing Research, under the Accelerated Research in Quantum Computing (ARQC) program.

\bibliography{ref.bib}

\begin{thebibliography}{45}%
\makeatletter
\providecommand \@ifxundefined [1]{%
 \@ifx{#1\undefined}
}%
\providecommand \@ifnum [1]{%
 \ifnum #1\expandafter \@firstoftwo
 \else \expandafter \@secondoftwo
 \fi
}%
\providecommand \@ifx [1]{%
 \ifx #1\expandafter \@firstoftwo
 \else \expandafter \@secondoftwo
 \fi
}%
\providecommand \natexlab [1]{#1}%
\providecommand \enquote  [1]{``#1''}%
\providecommand \bibnamefont  [1]{#1}%
\providecommand \bibfnamefont [1]{#1}%
\providecommand \citenamefont [1]{#1}%
\providecommand \href@noop [0]{\@secondoftwo}%
\providecommand \href [0]{\begingroup \@sanitize@url \@href}%
\providecommand \@href[1]{\@@startlink{#1}\@@href}%
\providecommand \@@href[1]{\endgroup#1\@@endlink}%
\providecommand \@sanitize@url [0]{\catcode `\\12\catcode `\$12\catcode
  `\&12\catcode `\#12\catcode `\^12\catcode `\_12\catcode `\%12\relax}%
\providecommand \@@startlink[1]{}%
\providecommand \@@endlink[0]{}%
\providecommand \url  [0]{\begingroup\@sanitize@url \@url }%
\providecommand \@url [1]{\endgroup\@href {#1}{\urlprefix }}%
\providecommand \urlprefix  [0]{URL }%
\providecommand \Eprint [0]{\href }%
\providecommand \doibase [0]{http://dx.doi.org/}%
\providecommand \selectlanguage [0]{\@gobble}%
\providecommand \bibinfo  [0]{\@secondoftwo}%
\providecommand \bibfield  [0]{\@secondoftwo}%
\providecommand \translation [1]{[#1]}%
\providecommand \BibitemOpen [0]{}%
\providecommand \bibitemStop [0]{}%
\providecommand \bibitemNoStop [0]{.\EOS\space}%
\providecommand \EOS [0]{\spacefactor3000\relax}%
\providecommand \BibitemShut  [1]{\csname bibitem#1\endcsname}%
\let\auto@bib@innerbib\@empty
\bibitem [{\citenamefont {Cerezo}\ \emph
  {et~al.}(2020{\natexlab{a}})\citenamefont {Cerezo}, \citenamefont
  {Arrasmith}, \citenamefont {Babbush}, \citenamefont {Benjamin}, \citenamefont
  {Endo}, \citenamefont {Fujii}, \citenamefont {McClean}, \citenamefont
  {Mitarai}, \citenamefont {Yuan}, \citenamefont {Cincio},\ and\ \citenamefont
  {Coles}}]{cerezo2020VQAreview}%
  \BibitemOpen
  \bibfield  {author} {\bibinfo {author} {\bibfnamefont {M.}~\bibnamefont
  {Cerezo}}, \bibinfo {author} {\bibfnamefont {Andrew}\ \bibnamefont
  {Arrasmith}}, \bibinfo {author} {\bibfnamefont {Ryan}\ \bibnamefont
  {Babbush}}, \bibinfo {author} {\bibfnamefont {Simon~C.}\ \bibnamefont
  {Benjamin}}, \bibinfo {author} {\bibfnamefont {Suguru}\ \bibnamefont {Endo}},
  \bibinfo {author} {\bibfnamefont {Keisuke}\ \bibnamefont {Fujii}}, \bibinfo
  {author} {\bibfnamefont {Jarrod~R.}\ \bibnamefont {McClean}}, \bibinfo
  {author} {\bibfnamefont {Kosuke}\ \bibnamefont {Mitarai}}, \bibinfo {author}
  {\bibfnamefont {Xiao}\ \bibnamefont {Yuan}}, \bibinfo {author} {\bibfnamefont
  {Lukasz}\ \bibnamefont {Cincio}}, \ and\ \bibinfo {author} {\bibfnamefont
  {Patrick~J.}\ \bibnamefont {Coles}},\ }\bibfield  {title} {\enquote {\bibinfo
  {title} {Variational quantum algorithms},}\ }\href
  {https://arxiv.org/abs/2012.09265} {\bibfield  {journal} {\bibinfo  {journal}
  {arXiv preprint arXiv:2012.09265}\ } (\bibinfo {year}
  {2020}{\natexlab{a}})}\BibitemShut {NoStop}%
\bibitem [{\citenamefont {{Peruzzo}}\ \emph {et~al.}(2014)\citenamefont
  {{Peruzzo}}, \citenamefont {{McClean}}, \citenamefont {{Shadbolt}},
  \citenamefont {{Yung}}, \citenamefont {{Zhou}}, \citenamefont {{Love}},
  \citenamefont {{Aspuru-Guzik}},\ and\ \citenamefont {{O'Brien}}}]{VQE}%
  \BibitemOpen
  \bibfield  {author} {\bibinfo {author} {\bibfnamefont {A.}~\bibnamefont
  {{Peruzzo}}}, \bibinfo {author} {\bibfnamefont {J.}~\bibnamefont
  {{McClean}}}, \bibinfo {author} {\bibfnamefont {P.}~\bibnamefont
  {{Shadbolt}}}, \bibinfo {author} {\bibfnamefont {M.-H.}\ \bibnamefont
  {{Yung}}}, \bibinfo {author} {\bibfnamefont {X.-Q.}\ \bibnamefont {{Zhou}}},
  \bibinfo {author} {\bibfnamefont {P.~J.}\ \bibnamefont {{Love}}}, \bibinfo
  {author} {\bibfnamefont {A.}~\bibnamefont {{Aspuru-Guzik}}}, \ and\ \bibinfo
  {author} {\bibfnamefont {J.~L.}\ \bibnamefont {{O'Brien}}},\ }\bibfield
  {title} {\enquote {\bibinfo {title} {{A variational eigenvalue solver on a
  photonic quantum processor}},}\ }\href {\doibase 10.1038/ncomms5213}
  {\bibfield  {journal} {\bibinfo  {journal} {Nature Communications}\ }\textbf
  {\bibinfo {volume} {5}},\ \bibinfo {eid} {4213} (\bibinfo {year}
  {2014})}\BibitemShut {NoStop}%
\bibitem [{\citenamefont {McClean}\ \emph {et~al.}(2016)\citenamefont
  {McClean}, \citenamefont {Romero}, \citenamefont {Babbush},\ and\
  \citenamefont {Aspuru-Guzik}}]{mcclean2016theory}%
  \BibitemOpen
  \bibfield  {author} {\bibinfo {author} {\bibfnamefont {Jarrod~R}\
  \bibnamefont {McClean}}, \bibinfo {author} {\bibfnamefont {Jonathan}\
  \bibnamefont {Romero}}, \bibinfo {author} {\bibfnamefont {Ryan}\ \bibnamefont
  {Babbush}}, \ and\ \bibinfo {author} {\bibfnamefont {Al{\'a}n}\ \bibnamefont
  {Aspuru-Guzik}},\ }\bibfield  {title} {\enquote {\bibinfo {title} {The theory
  of variational hybrid quantum-classical algorithms},}\ }\href
  {https://iopscience.iop.org/article/10.1088/1367-2630/18/2/023023/meta}
  {\bibfield  {journal} {\bibinfo  {journal} {New Journal of Physics}\ }\textbf
  {\bibinfo {volume} {18}},\ \bibinfo {pages} {023023} (\bibinfo {year}
  {2016})}\BibitemShut {NoStop}%
\bibitem [{\citenamefont {Farhi}\ \emph {et~al.}()\citenamefont {Farhi},
  \citenamefont {Goldstone},\ and\ \citenamefont {Gutmann}}]{qaoa2014}%
  \BibitemOpen
  \bibfield  {author} {\bibinfo {author} {\bibfnamefont {E.}~\bibnamefont
  {Farhi}}, \bibinfo {author} {\bibfnamefont {J.}~\bibnamefont {Goldstone}}, \
  and\ \bibinfo {author} {\bibfnamefont {S.}~\bibnamefont {Gutmann}},\
  }\bibfield  {title} {\enquote {\bibinfo {title} {A quantum approximate
  optimization algorithm},}\ }\href@noop {} {\ }\Eprint
  {http://arxiv.org/abs/1411.4028} {arXiv:1411.4028 [quant-ph]} \BibitemShut
  {NoStop}%
\bibitem [{\citenamefont {{Romero}}\ \emph {et~al.}(2017)\citenamefont
  {{Romero}}, \citenamefont {{Olson}},\ and\ \citenamefont
  {{Aspuru-Guzik}}}]{Romero}%
  \BibitemOpen
  \bibfield  {author} {\bibinfo {author} {\bibfnamefont {J.}~\bibnamefont
  {{Romero}}}, \bibinfo {author} {\bibfnamefont {J.~P.}\ \bibnamefont
  {{Olson}}}, \ and\ \bibinfo {author} {\bibfnamefont {A.}~\bibnamefont
  {{Aspuru-Guzik}}},\ }\bibfield  {title} {\enquote {\bibinfo {title} {{Quantum
  autoencoders for efficient compression of quantum data}},}\ }\href {\doibase
  10.1088/2058-9565/aa8072} {\bibfield  {journal} {\bibinfo  {journal} {Quantum
  Science and Technology}\ }\textbf {\bibinfo {volume} {2}},\ \bibinfo {pages}
  {045001} (\bibinfo {year} {2017})}\BibitemShut {NoStop}%
\bibitem [{\citenamefont {Khatri}\ \emph {et~al.}(2019)\citenamefont {Khatri},
  \citenamefont {LaRose}, \citenamefont {Poremba}, \citenamefont {Cincio},
  \citenamefont {Sornborger},\ and\ \citenamefont {Coles}}]{QAQC}%
  \BibitemOpen
  \bibfield  {author} {\bibinfo {author} {\bibfnamefont {S.}~\bibnamefont
  {Khatri}}, \bibinfo {author} {\bibfnamefont {R.}~\bibnamefont {LaRose}},
  \bibinfo {author} {\bibfnamefont {A.}~\bibnamefont {Poremba}}, \bibinfo
  {author} {\bibfnamefont {L.}~\bibnamefont {Cincio}}, \bibinfo {author}
  {\bibfnamefont {A.~T.}\ \bibnamefont {Sornborger}}, \ and\ \bibinfo {author}
  {\bibfnamefont {P.~J.}\ \bibnamefont {Coles}},\ }\bibfield  {title} {\enquote
  {\bibinfo {title} {Quantum-assisted quantum compiling},}\ }\href {\doibase
  10.22331/q-2019-05-13-140} {\bibfield  {journal} {\bibinfo  {journal}
  {{Quantum}}\ }\textbf {\bibinfo {volume} {3}},\ \bibinfo {pages} {140}
  (\bibinfo {year} {2019})}\BibitemShut {NoStop}%
\bibitem [{\citenamefont {LaRose}\ \emph {et~al.}(2018)\citenamefont {LaRose},
  \citenamefont {Tikku}, \citenamefont {O'Neel-Judy}, \citenamefont {Cincio},\
  and\ \citenamefont {Coles}}]{VQSD}%
  \BibitemOpen
  \bibfield  {author} {\bibinfo {author} {\bibfnamefont {R.}~\bibnamefont
  {LaRose}}, \bibinfo {author} {\bibfnamefont {A.}~\bibnamefont {Tikku}},
  \bibinfo {author} {\bibfnamefont {{\'E}.}~\bibnamefont {O'Neel-Judy}},
  \bibinfo {author} {\bibfnamefont {L.}~\bibnamefont {Cincio}}, \ and\ \bibinfo
  {author} {\bibfnamefont {P.~J.}\ \bibnamefont {Coles}},\ }\bibfield  {title}
  {\enquote {\bibinfo {title} {Variational quantum state diagonalization},}\
  }\href {\doibase 10.1038/s41534-019-0167-6} {\bibfield  {journal} {\bibinfo
  {journal} {npj Quantum Information}\ }\textbf {\bibinfo {volume} {5}},\
  \bibinfo {pages} {1--10} (\bibinfo {year} {2018})}\BibitemShut {NoStop}%
\bibitem [{\citenamefont {Arrasmith}\ \emph {et~al.}(2019)\citenamefont
  {Arrasmith}, \citenamefont {Cincio}, \citenamefont {Sornborger},
  \citenamefont {Zurek},\ and\ \citenamefont
  {Coles}}]{arrasmith2019variational}%
  \BibitemOpen
  \bibfield  {author} {\bibinfo {author} {\bibfnamefont {A.}~\bibnamefont
  {Arrasmith}}, \bibinfo {author} {\bibfnamefont {L.}~\bibnamefont {Cincio}},
  \bibinfo {author} {\bibfnamefont {A.~T.}\ \bibnamefont {Sornborger}},
  \bibinfo {author} {\bibfnamefont {W.~H.}\ \bibnamefont {Zurek}}, \ and\
  \bibinfo {author} {\bibfnamefont {P.~J.}\ \bibnamefont {Coles}},\ }\bibfield
  {title} {\enquote {\bibinfo {title} {Variational consistent histories as a
  hybrid algorithm for quantum foundations},}\ }\href {\doibase
  10.1038/s41467-019-11417-0} {\bibfield  {journal} {\bibinfo  {journal}
  {Nature communications}\ }\textbf {\bibinfo {volume} {10}},\ \bibinfo {pages}
  {3438} (\bibinfo {year} {2019})}\BibitemShut {NoStop}%
\bibitem [{\citenamefont {Cerezo}\ \emph
  {et~al.}(2020{\natexlab{b}})\citenamefont {Cerezo}, \citenamefont {Poremba},
  \citenamefont {Cincio},\ and\ \citenamefont
  {Coles}}]{cerezo2020variationalfidelity}%
  \BibitemOpen
  \bibfield  {author} {\bibinfo {author} {\bibfnamefont {Marco}\ \bibnamefont
  {Cerezo}}, \bibinfo {author} {\bibfnamefont {Alexander}\ \bibnamefont
  {Poremba}}, \bibinfo {author} {\bibfnamefont {Lukasz}\ \bibnamefont
  {Cincio}}, \ and\ \bibinfo {author} {\bibfnamefont {Patrick~J}\ \bibnamefont
  {Coles}},\ }\bibfield  {title} {\enquote {\bibinfo {title} {Variational
  quantum fidelity estimation},}\ }\href
  {https://quantum-journal.org/papers/q-2020-03-26-248/} {\bibfield  {journal}
  {\bibinfo  {journal} {Quantum}\ }\textbf {\bibinfo {volume} {4}},\ \bibinfo
  {pages} {248} (\bibinfo {year} {2020}{\natexlab{b}})}\BibitemShut {NoStop}%
\bibitem [{\citenamefont {Cirstoiu}\ \emph {et~al.}(2020)\citenamefont
  {Cirstoiu}, \citenamefont {Holmes}, \citenamefont {Iosue}, \citenamefont
  {Cincio}, \citenamefont {Coles},\ and\ \citenamefont
  {Sornborger}}]{cirstoiu2019variational}%
  \BibitemOpen
  \bibfield  {author} {\bibinfo {author} {\bibfnamefont {Cristina}\
  \bibnamefont {Cirstoiu}}, \bibinfo {author} {\bibfnamefont {Zoe}\
  \bibnamefont {Holmes}}, \bibinfo {author} {\bibfnamefont {Joseph}\
  \bibnamefont {Iosue}}, \bibinfo {author} {\bibfnamefont {Lukasz}\
  \bibnamefont {Cincio}}, \bibinfo {author} {\bibfnamefont {Patrick~J}\
  \bibnamefont {Coles}}, \ and\ \bibinfo {author} {\bibfnamefont {Andrew}\
  \bibnamefont {Sornborger}},\ }\bibfield  {title} {\enquote {\bibinfo {title}
  {Variational fast forwarding for quantum simulation beyond the coherence
  time},}\ }\href {https://www.nature.com/articles/s41534-020-00302-0}
  {\bibfield  {journal} {\bibinfo  {journal} {npj Quantum Information}\
  }\textbf {\bibinfo {volume} {6}},\ \bibinfo {pages} {1--10} (\bibinfo {year}
  {2020})}\BibitemShut {NoStop}%
\bibitem [{\citenamefont {Bravo-Prieto}\ \emph {et~al.}(2019)\citenamefont
  {Bravo-Prieto}, \citenamefont {LaRose}, \citenamefont {Cerezo}, \citenamefont
  {Subasi}, \citenamefont {Cincio},\ and\ \citenamefont
  {Coles}}]{bravo-prieto2019}%
  \BibitemOpen
  \bibfield  {author} {\bibinfo {author} {\bibfnamefont {Carlos}\ \bibnamefont
  {Bravo-Prieto}}, \bibinfo {author} {\bibfnamefont {Ryan}\ \bibnamefont
  {LaRose}}, \bibinfo {author} {\bibfnamefont {M.}~\bibnamefont {Cerezo}},
  \bibinfo {author} {\bibfnamefont {Yigit}\ \bibnamefont {Subasi}}, \bibinfo
  {author} {\bibfnamefont {Lukasz}\ \bibnamefont {Cincio}}, \ and\ \bibinfo
  {author} {\bibfnamefont {Patrick~J.}\ \bibnamefont {Coles}},\ }\bibfield
  {title} {\enquote {\bibinfo {title} {Variational quantum linear solver: A
  hybrid algorithm for linear systems},}\ }\href
  {https://arxiv.org/abs/1909.05820} {\bibfield  {journal} {\bibinfo  {journal}
  {arXiv:1909.05820}\ } (\bibinfo {year} {2019})}\BibitemShut {NoStop}%
\bibitem [{\citenamefont {Xu}\ \emph {et~al.}(2019)\citenamefont {Xu},
  \citenamefont {Sun}, \citenamefont {Endo}, \citenamefont {Li}, \citenamefont
  {Benjamin},\ and\ \citenamefont {Yuan}}]{xu2019variational}%
  \BibitemOpen
  \bibfield  {author} {\bibinfo {author} {\bibfnamefont {Xiaosi}\ \bibnamefont
  {Xu}}, \bibinfo {author} {\bibfnamefont {Jinzhao}\ \bibnamefont {Sun}},
  \bibinfo {author} {\bibfnamefont {Suguru}\ \bibnamefont {Endo}}, \bibinfo
  {author} {\bibfnamefont {Ying}\ \bibnamefont {Li}}, \bibinfo {author}
  {\bibfnamefont {Simon~C}\ \bibnamefont {Benjamin}}, \ and\ \bibinfo {author}
  {\bibfnamefont {Xiao}\ \bibnamefont {Yuan}},\ }\bibfield  {title} {\enquote
  {\bibinfo {title} {Variational algorithms for linear algebra},}\ }\href
  {https://arxiv.org/abs/1909.03898} {\bibfield  {journal} {\bibinfo  {journal}
  {arXiv preprint arXiv:1909.03898}\ } (\bibinfo {year} {2019})}\BibitemShut
  {NoStop}%
\bibitem [{\citenamefont {Cerezo}\ \emph
  {et~al.}(2020{\natexlab{c}})\citenamefont {Cerezo}, \citenamefont {Sharma},
  \citenamefont {Arrasmith},\ and\ \citenamefont
  {Coles}}]{cerezo2020variational}%
  \BibitemOpen
  \bibfield  {author} {\bibinfo {author} {\bibfnamefont {M}~\bibnamefont
  {Cerezo}}, \bibinfo {author} {\bibfnamefont {Kunal}\ \bibnamefont {Sharma}},
  \bibinfo {author} {\bibfnamefont {Andrew}\ \bibnamefont {Arrasmith}}, \ and\
  \bibinfo {author} {\bibfnamefont {Patrick~J}\ \bibnamefont {Coles}},\
  }\bibfield  {title} {\enquote {\bibinfo {title} {Variational quantum state
  eigensolver},}\ }\href {https://arxiv.org/abs/2004.01372} {\bibfield
  {journal} {\bibinfo  {journal} {arXiv preprint arXiv:2004.01372}\ } (\bibinfo
  {year} {2020}{\natexlab{c}})}\BibitemShut {NoStop}%
\bibitem [{\citenamefont {Kyriienko}\ \emph {et~al.}(2020)\citenamefont
  {Kyriienko}, \citenamefont {Paine},\ and\ \citenamefont
  {Elfving}}]{kyriienko2020solving}%
  \BibitemOpen
  \bibfield  {author} {\bibinfo {author} {\bibfnamefont {Oleksandr}\
  \bibnamefont {Kyriienko}}, \bibinfo {author} {\bibfnamefont {Annie~E}\
  \bibnamefont {Paine}}, \ and\ \bibinfo {author} {\bibfnamefont {Vincent~E}\
  \bibnamefont {Elfving}},\ }\bibfield  {title} {\enquote {\bibinfo {title}
  {Solving nonlinear differential equations with differentiable quantum
  circuits},}\ }\href {https://arxiv.org/abs/2011.10395} {\bibfield  {journal}
  {\bibinfo  {journal} {arXiv preprint arXiv:2011.10395}\ } (\bibinfo {year}
  {2020})}\BibitemShut {NoStop}%
\bibitem [{\citenamefont {Chivilikhin}\ \emph {et~al.}(2020)\citenamefont
  {Chivilikhin}, \citenamefont {Samarin}, \citenamefont {Ulyantsev},
  \citenamefont {Iorsh}, \citenamefont {Oganov},\ and\ \citenamefont
  {Kyriienko}}]{chivilikhin2020mog}%
  \BibitemOpen
  \bibfield  {author} {\bibinfo {author} {\bibfnamefont {D}~\bibnamefont
  {Chivilikhin}}, \bibinfo {author} {\bibfnamefont {A}~\bibnamefont {Samarin}},
  \bibinfo {author} {\bibfnamefont {V}~\bibnamefont {Ulyantsev}}, \bibinfo
  {author} {\bibfnamefont {I}~\bibnamefont {Iorsh}}, \bibinfo {author}
  {\bibfnamefont {AR}~\bibnamefont {Oganov}}, \ and\ \bibinfo {author}
  {\bibfnamefont {O}~\bibnamefont {Kyriienko}},\ }\bibfield  {title} {\enquote
  {\bibinfo {title} {Mog-vqe: Multiobjective genetic variational quantum
  eigensolver},}\ }\href {https://arxiv.org/abs/2007.04424} {\bibfield
  {journal} {\bibinfo  {journal} {arXiv preprint arXiv:2007.04424}\ } (\bibinfo
  {year} {2020})}\BibitemShut {NoStop}%
\bibitem [{\citenamefont {Sharma}\ \emph
  {et~al.}(2020{\natexlab{a}})\citenamefont {Sharma}, \citenamefont {Khatri},
  \citenamefont {Cerezo},\ and\ \citenamefont {Coles}}]{sharma2019noise}%
  \BibitemOpen
  \bibfield  {author} {\bibinfo {author} {\bibfnamefont {Kunal}\ \bibnamefont
  {Sharma}}, \bibinfo {author} {\bibfnamefont {Sumeet}\ \bibnamefont {Khatri}},
  \bibinfo {author} {\bibfnamefont {M.}~\bibnamefont {Cerezo}}, \ and\ \bibinfo
  {author} {\bibfnamefont {Patrick~J}\ \bibnamefont {Coles}},\ }\bibfield
  {title} {\enquote {\bibinfo {title} {Noise resilience of variational quantum
  compiling},}\ }\href
  {https://iopscience.iop.org/article/10.1088/1367-2630/ab784c} {\bibfield
  {journal} {\bibinfo  {journal} {New Journal of Physics}\ }\textbf {\bibinfo
  {volume} {22}},\ \bibinfo {pages} {043006} (\bibinfo {year}
  {2020}{\natexlab{a}})}\BibitemShut {NoStop}%
\bibitem [{\citenamefont {Schuld}\ \emph {et~al.}(2014)\citenamefont {Schuld},
  \citenamefont {Sinayskiy},\ and\ \citenamefont
  {Petruccione}}]{schuld2014quest}%
  \BibitemOpen
  \bibfield  {author} {\bibinfo {author} {\bibfnamefont {Maria}\ \bibnamefont
  {Schuld}}, \bibinfo {author} {\bibfnamefont {Ilya}\ \bibnamefont
  {Sinayskiy}}, \ and\ \bibinfo {author} {\bibfnamefont {Francesco}\
  \bibnamefont {Petruccione}},\ }\bibfield  {title} {\enquote {\bibinfo {title}
  {The quest for a quantum neural network},}\ }\href
  {https://link.springer.com/article/10.1007/s11128-014-0809-8} {\bibfield
  {journal} {\bibinfo  {journal} {Quantum Information Processing}\ }\textbf
  {\bibinfo {volume} {13}},\ \bibinfo {pages} {2567--2586} (\bibinfo {year}
  {2014})}\BibitemShut {NoStop}%
\bibitem [{\citenamefont {Cong}\ \emph {et~al.}(2019)\citenamefont {Cong},
  \citenamefont {Choi},\ and\ \citenamefont {Lukin}}]{cong2019quantum}%
  \BibitemOpen
  \bibfield  {author} {\bibinfo {author} {\bibfnamefont {Iris}\ \bibnamefont
  {Cong}}, \bibinfo {author} {\bibfnamefont {Soonwon}\ \bibnamefont {Choi}}, \
  and\ \bibinfo {author} {\bibfnamefont {Mikhail~D}\ \bibnamefont {Lukin}},\
  }\bibfield  {title} {\enquote {\bibinfo {title} {Quantum convolutional neural
  networks},}\ }\href {https://www.nature.com/articles/s41567-019-0648-8}
  {\bibfield  {journal} {\bibinfo  {journal} {Nature Physics}\ }\textbf
  {\bibinfo {volume} {15}},\ \bibinfo {pages} {1273--1278} (\bibinfo {year}
  {2019})}\BibitemShut {NoStop}%
\bibitem [{\citenamefont {Beer}\ \emph {et~al.}(2020)\citenamefont {Beer},
  \citenamefont {Bondarenko}, \citenamefont {Farrelly}, \citenamefont
  {Osborne}, \citenamefont {Salzmann}, \citenamefont {Scheiermann},\ and\
  \citenamefont {Wolf}}]{beer2020training}%
  \BibitemOpen
  \bibfield  {author} {\bibinfo {author} {\bibfnamefont {Kerstin}\ \bibnamefont
  {Beer}}, \bibinfo {author} {\bibfnamefont {Dmytro}\ \bibnamefont
  {Bondarenko}}, \bibinfo {author} {\bibfnamefont {Terry}\ \bibnamefont
  {Farrelly}}, \bibinfo {author} {\bibfnamefont {Tobias~J}\ \bibnamefont
  {Osborne}}, \bibinfo {author} {\bibfnamefont {Robert}\ \bibnamefont
  {Salzmann}}, \bibinfo {author} {\bibfnamefont {Daniel}\ \bibnamefont
  {Scheiermann}}, \ and\ \bibinfo {author} {\bibfnamefont {Ramona}\
  \bibnamefont {Wolf}},\ }\bibfield  {title} {\enquote {\bibinfo {title}
  {Training deep quantum neural networks},}\ }\href
  {https://www.nature.com/articles/s41467-020-14454-2} {\bibfield  {journal}
  {\bibinfo  {journal} {Nature Communications}\ }\textbf {\bibinfo {volume}
  {11}},\ \bibinfo {pages} {1--6} (\bibinfo {year} {2020})}\BibitemShut
  {NoStop}%
\bibitem [{\citenamefont {Verdon}\ \emph {et~al.}(2018)\citenamefont {Verdon},
  \citenamefont {Pye},\ and\ \citenamefont {Broughton}}]{verdon2018universal}%
  \BibitemOpen
  \bibfield  {author} {\bibinfo {author} {\bibfnamefont {Guillaume}\
  \bibnamefont {Verdon}}, \bibinfo {author} {\bibfnamefont {Jason}\
  \bibnamefont {Pye}}, \ and\ \bibinfo {author} {\bibfnamefont {Michael}\
  \bibnamefont {Broughton}},\ }\bibfield  {title} {\enquote {\bibinfo {title}
  {A universal training algorithm for quantum deep learning},}\ }\href
  {https://arxiv.org/abs/1806.09729} {\bibfield  {journal} {\bibinfo  {journal}
  {arXiv preprint arXiv:1806.09729}\ } (\bibinfo {year} {2018})}\BibitemShut
  {NoStop}%
\bibitem [{\citenamefont {Abbas}\ \emph {et~al.}(2020)\citenamefont {Abbas},
  \citenamefont {Sutter}, \citenamefont {Zoufal}, \citenamefont {Lucchi},
  \citenamefont {Figalli},\ and\ \citenamefont {Woerner}}]{abbas2020power}%
  \BibitemOpen
  \bibfield  {author} {\bibinfo {author} {\bibfnamefont {Amira}\ \bibnamefont
  {Abbas}}, \bibinfo {author} {\bibfnamefont {David}\ \bibnamefont {Sutter}},
  \bibinfo {author} {\bibfnamefont {Christa}\ \bibnamefont {Zoufal}}, \bibinfo
  {author} {\bibfnamefont {Aur{\'e}lien}\ \bibnamefont {Lucchi}}, \bibinfo
  {author} {\bibfnamefont {Alessio}\ \bibnamefont {Figalli}}, \ and\ \bibinfo
  {author} {\bibfnamefont {Stefan}\ \bibnamefont {Woerner}},\ }\bibfield
  {title} {\enquote {\bibinfo {title} {The power of quantum neural networks},}\
  }\href {https://arxiv.org/abs/2011.00027} {\bibfield  {journal} {\bibinfo
  {journal} {arXiv preprint arXiv:2011.00027}\ } (\bibinfo {year}
  {2020})}\BibitemShut {NoStop}%
\bibitem [{\citenamefont {McClean}\ \emph {et~al.}(2018)\citenamefont
  {McClean}, \citenamefont {Boixo}, \citenamefont {Smelyanskiy}, \citenamefont
  {Babbush},\ and\ \citenamefont {Neven}}]{mcclean2018barren}%
  \BibitemOpen
  \bibfield  {author} {\bibinfo {author} {\bibfnamefont {Jarrod~R}\
  \bibnamefont {McClean}}, \bibinfo {author} {\bibfnamefont {Sergio}\
  \bibnamefont {Boixo}}, \bibinfo {author} {\bibfnamefont {Vadim~N}\
  \bibnamefont {Smelyanskiy}}, \bibinfo {author} {\bibfnamefont {Ryan}\
  \bibnamefont {Babbush}}, \ and\ \bibinfo {author} {\bibfnamefont {Hartmut}\
  \bibnamefont {Neven}},\ }\bibfield  {title} {\enquote {\bibinfo {title}
  {Barren plateaus in quantum neural network training landscapes},}\ }\href
  {https://www.nature.com/articles/s41467-018-07090-4} {\bibfield  {journal}
  {\bibinfo  {journal} {Nature communications}\ }\textbf {\bibinfo {volume}
  {9}},\ \bibinfo {pages} {4812} (\bibinfo {year} {2018})}\BibitemShut
  {NoStop}%
\bibitem [{\citenamefont {Cerezo}\ \emph {et~al.}(2021)\citenamefont {Cerezo},
  \citenamefont {Sone}, \citenamefont {Volkoff}, \citenamefont {Cincio},\ and\
  \citenamefont {Coles}}]{cerezo2020cost}%
  \BibitemOpen
  \bibfield  {author} {\bibinfo {author} {\bibfnamefont {M.}~\bibnamefont
  {Cerezo}}, \bibinfo {author} {\bibfnamefont {Akira}\ \bibnamefont {Sone}},
  \bibinfo {author} {\bibfnamefont {Tyler}\ \bibnamefont {Volkoff}}, \bibinfo
  {author} {\bibfnamefont {Lukasz}\ \bibnamefont {Cincio}}, \ and\ \bibinfo
  {author} {\bibfnamefont {Patrick~J}\ \bibnamefont {Coles}},\ }\bibfield
  {title} {\enquote {\bibinfo {title} {Cost function dependent barren plateaus
  in shallow parametrized quantum circuits},}\ }\href {\doibase
  10.1038/s41467-021-21728-w} {\bibfield  {journal} {\bibinfo  {journal}
  {Nature Communications}\ }\textbf {\bibinfo {volume} {12}},\ \bibinfo {pages}
  {1791} (\bibinfo {year} {2021})}\BibitemShut {NoStop}%
\bibitem [{\citenamefont {Sharma}\ \emph
  {et~al.}(2020{\natexlab{b}})\citenamefont {Sharma}, \citenamefont {Cerezo},
  \citenamefont {Cincio},\ and\ \citenamefont
  {Coles}}]{sharma2020trainability}%
  \BibitemOpen
  \bibfield  {author} {\bibinfo {author} {\bibfnamefont {Kunal}\ \bibnamefont
  {Sharma}}, \bibinfo {author} {\bibfnamefont {M}~\bibnamefont {Cerezo}},
  \bibinfo {author} {\bibfnamefont {Lukasz}\ \bibnamefont {Cincio}}, \ and\
  \bibinfo {author} {\bibfnamefont {Patrick~J}\ \bibnamefont {Coles}},\
  }\bibfield  {title} {\enquote {\bibinfo {title} {Trainability of dissipative
  perceptron-based quantum neural networks},}\ }\href
  {https://arxiv.org/abs/2005.12458} {\bibfield  {journal} {\bibinfo  {journal}
  {arXiv preprint arXiv:2005.12458}\ } (\bibinfo {year}
  {2020}{\natexlab{b}})}\BibitemShut {NoStop}%
\bibitem [{\citenamefont {Wang}\ \emph {et~al.}(2020)\citenamefont {Wang},
  \citenamefont {Fontana}, \citenamefont {Cerezo}, \citenamefont {Sharma},
  \citenamefont {Sone}, \citenamefont {Cincio},\ and\ \citenamefont
  {Coles}}]{wang2020noise}%
  \BibitemOpen
  \bibfield  {author} {\bibinfo {author} {\bibfnamefont {Samson}\ \bibnamefont
  {Wang}}, \bibinfo {author} {\bibfnamefont {Enrico}\ \bibnamefont {Fontana}},
  \bibinfo {author} {\bibfnamefont {M}~\bibnamefont {Cerezo}}, \bibinfo
  {author} {\bibfnamefont {Kunal}\ \bibnamefont {Sharma}}, \bibinfo {author}
  {\bibfnamefont {Akira}\ \bibnamefont {Sone}}, \bibinfo {author}
  {\bibfnamefont {Lukasz}\ \bibnamefont {Cincio}}, \ and\ \bibinfo {author}
  {\bibfnamefont {Patrick~J}\ \bibnamefont {Coles}},\ }\bibfield  {title}
  {\enquote {\bibinfo {title} {Noise-induced barren plateaus in variational
  quantum algorithms},}\ }\href {https://arxiv.org/abs/2007.14384} {\bibfield
  {journal} {\bibinfo  {journal} {arXiv preprint arXiv:2007.14384}\ } (\bibinfo
  {year} {2020})}\BibitemShut {NoStop}%
\bibitem [{\citenamefont {Holmes}\ \emph {et~al.}(2020)\citenamefont {Holmes},
  \citenamefont {Arrasmith}, \citenamefont {Yan}, \citenamefont {Coles},
  \citenamefont {Albrecht},\ and\ \citenamefont
  {Sornborger}}]{holmes2020barren}%
  \BibitemOpen
  \bibfield  {author} {\bibinfo {author} {\bibfnamefont {Zo{\"e}}\ \bibnamefont
  {Holmes}}, \bibinfo {author} {\bibfnamefont {Andrew}\ \bibnamefont
  {Arrasmith}}, \bibinfo {author} {\bibfnamefont {Bin}\ \bibnamefont {Yan}},
  \bibinfo {author} {\bibfnamefont {Patrick~J}\ \bibnamefont {Coles}}, \bibinfo
  {author} {\bibfnamefont {Andreas}\ \bibnamefont {Albrecht}}, \ and\ \bibinfo
  {author} {\bibfnamefont {Andrew~T}\ \bibnamefont {Sornborger}},\ }\bibfield
  {title} {\enquote {\bibinfo {title} {Barren plateaus preclude learning
  scramblers},}\ }\href {https://arxiv.org/abs/2009.14808} {\bibfield
  {journal} {\bibinfo  {journal} {arXiv preprint arXiv:2009.14808}\ } (\bibinfo
  {year} {2020})}\BibitemShut {NoStop}%
\bibitem [{\citenamefont {Marrero}\ \emph {et~al.}(2020)\citenamefont
  {Marrero}, \citenamefont {Kieferov{\'a}},\ and\ \citenamefont
  {Wiebe}}]{marrero2020entanglement}%
  \BibitemOpen
  \bibfield  {author} {\bibinfo {author} {\bibfnamefont {Carlos~Ortiz}\
  \bibnamefont {Marrero}}, \bibinfo {author} {\bibfnamefont {M{\'a}ria}\
  \bibnamefont {Kieferov{\'a}}}, \ and\ \bibinfo {author} {\bibfnamefont
  {Nathan}\ \bibnamefont {Wiebe}},\ }\bibfield  {title} {\enquote {\bibinfo
  {title} {Entanglement induced barren plateaus},}\ }\href
  {https://arxiv.org/abs/2010.15968} {\bibfield  {journal} {\bibinfo  {journal}
  {arXiv preprint arXiv:2010.15968}\ } (\bibinfo {year} {2020})}\BibitemShut
  {NoStop}%
\bibitem [{\citenamefont {Uvarov}\ and\ \citenamefont
  {Biamonte}(2020)}]{uvarov2020barren}%
  \BibitemOpen
  \bibfield  {author} {\bibinfo {author} {\bibfnamefont {Alexey}\ \bibnamefont
  {Uvarov}}\ and\ \bibinfo {author} {\bibfnamefont {Jacob}\ \bibnamefont
  {Biamonte}},\ }\bibfield  {title} {\enquote {\bibinfo {title} {On barren
  plateaus and cost function locality in variational quantum algorithms},}\
  }\href {https://arxiv.org/abs/2011.10530} {\bibfield  {journal} {\bibinfo
  {journal} {arXiv preprint arXiv:2011.10530}\ } (\bibinfo {year}
  {2020})}\BibitemShut {NoStop}%
\bibitem [{\citenamefont {Arrasmith}\ \emph {et~al.}(2020)\citenamefont
  {Arrasmith}, \citenamefont {Cerezo}, \citenamefont {Czarnik}, \citenamefont
  {Cincio},\ and\ \citenamefont {Coles}}]{arrasmith2020effect}%
  \BibitemOpen
  \bibfield  {author} {\bibinfo {author} {\bibfnamefont {Andrew}\ \bibnamefont
  {Arrasmith}}, \bibinfo {author} {\bibfnamefont {M}~\bibnamefont {Cerezo}},
  \bibinfo {author} {\bibfnamefont {Piotr}\ \bibnamefont {Czarnik}}, \bibinfo
  {author} {\bibfnamefont {Lukasz}\ \bibnamefont {Cincio}}, \ and\ \bibinfo
  {author} {\bibfnamefont {Patrick~J}\ \bibnamefont {Coles}},\ }\bibfield
  {title} {\enquote {\bibinfo {title} {Effect of barren plateaus on
  gradient-free optimization},}\ }\href {https://arxiv.org/abs/2011.12245}
  {\bibfield  {journal} {\bibinfo  {journal} {arXiv preprint arXiv:2011.12245}\
  } (\bibinfo {year} {2020})}\BibitemShut {NoStop}%
\bibitem [{\citenamefont {Holmes}\ \emph {et~al.}(2021)\citenamefont {Holmes},
  \citenamefont {Sharma}, \citenamefont {Cerezo},\ and\ \citenamefont
  {Coles}}]{holmes2021connecting}%
  \BibitemOpen
  \bibfield  {author} {\bibinfo {author} {\bibfnamefont {Zo{\"e}}\ \bibnamefont
  {Holmes}}, \bibinfo {author} {\bibfnamefont {Kunal}\ \bibnamefont {Sharma}},
  \bibinfo {author} {\bibfnamefont {M.}~\bibnamefont {Cerezo}}, \ and\ \bibinfo
  {author} {\bibfnamefont {Patrick~J.}\ \bibnamefont {Coles}},\ }\bibfield
  {title} {\enquote {\bibinfo {title} {Connecting ansatz expressibility to
  gradient magnitudes and barren plateaus},}\ }\href
  {https://arxiv.org/abs/2101.02138} {\bibfield  {journal} {\bibinfo  {journal}
  {arXiv preprint arXiv:2101.02138}\ } (\bibinfo {year} {2021})}\BibitemShut
  {NoStop}%
\bibitem [{\citenamefont {Patti}\ \emph {et~al.}(2020)\citenamefont {Patti},
  \citenamefont {Najafi}, \citenamefont {Gao},\ and\ \citenamefont
  {Yelin}}]{patti2020entanglement}%
  \BibitemOpen
  \bibfield  {author} {\bibinfo {author} {\bibfnamefont {Taylor~L}\
  \bibnamefont {Patti}}, \bibinfo {author} {\bibfnamefont {Khadijeh}\
  \bibnamefont {Najafi}}, \bibinfo {author} {\bibfnamefont {Xun}\ \bibnamefont
  {Gao}}, \ and\ \bibinfo {author} {\bibfnamefont {Susanne~F}\ \bibnamefont
  {Yelin}},\ }\bibfield  {title} {\enquote {\bibinfo {title} {Entanglement
  devised barren plateau mitigation},}\ }\href
  {https://arxiv.org/abs/2012.12658} {\bibfield  {journal} {\bibinfo  {journal}
  {arXiv preprint arXiv:2012.12658}\ } (\bibinfo {year} {2020})}\BibitemShut
  {NoStop}%
\bibitem [{\citenamefont {Verdon}\ \emph {et~al.}(2019)\citenamefont {Verdon},
  \citenamefont {Broughton}, \citenamefont {McClean}, \citenamefont {Sung},
  \citenamefont {Babbush}, \citenamefont {Jiang}, \citenamefont {Neven},\ and\
  \citenamefont {Mohseni}}]{verdon2019learning}%
  \BibitemOpen
  \bibfield  {author} {\bibinfo {author} {\bibfnamefont {Guillaume}\
  \bibnamefont {Verdon}}, \bibinfo {author} {\bibfnamefont {Michael}\
  \bibnamefont {Broughton}}, \bibinfo {author} {\bibfnamefont {Jarrod~R}\
  \bibnamefont {McClean}}, \bibinfo {author} {\bibfnamefont {Kevin~J}\
  \bibnamefont {Sung}}, \bibinfo {author} {\bibfnamefont {Ryan}\ \bibnamefont
  {Babbush}}, \bibinfo {author} {\bibfnamefont {Zhang}\ \bibnamefont {Jiang}},
  \bibinfo {author} {\bibfnamefont {Hartmut}\ \bibnamefont {Neven}}, \ and\
  \bibinfo {author} {\bibfnamefont {Masoud}\ \bibnamefont {Mohseni}},\
  }\bibfield  {title} {\enquote {\bibinfo {title} {Learning to learn with
  quantum neural networks via classical neural networks},}\ }\href
  {https://arxiv.org/abs/1907.05415} {\bibfield  {journal} {\bibinfo  {journal}
  {arXiv preprint arXiv:1907.05415}\ } (\bibinfo {year} {2019})}\BibitemShut
  {NoStop}%
\bibitem [{\citenamefont {Volkoff}\ and\ \citenamefont
  {Coles}(2021)}]{volkoff2021large}%
  \BibitemOpen
  \bibfield  {author} {\bibinfo {author} {\bibfnamefont {Tyler}\ \bibnamefont
  {Volkoff}}\ and\ \bibinfo {author} {\bibfnamefont {Patrick~J}\ \bibnamefont
  {Coles}},\ }\bibfield  {title} {\enquote {\bibinfo {title} {Large gradients
  via correlation in random parameterized quantum circuits},}\ }\href
  {https://iopscience.iop.org/article/10.1088/2058-9565/abd891} {\bibfield
  {journal} {\bibinfo  {journal} {Quantum Science and Technology}\ }\textbf
  {\bibinfo {volume} {6}},\ \bibinfo {pages} {025008} (\bibinfo {year}
  {2021})}\BibitemShut {NoStop}%
\bibitem [{\citenamefont {Skolik}\ \emph {et~al.}(2020)\citenamefont {Skolik},
  \citenamefont {McClean}, \citenamefont {Mohseni}, \citenamefont {van~der
  Smagt},\ and\ \citenamefont {Leib}}]{skolik2020layerwise}%
  \BibitemOpen
  \bibfield  {author} {\bibinfo {author} {\bibfnamefont {Andrea}\ \bibnamefont
  {Skolik}}, \bibinfo {author} {\bibfnamefont {Jarrod~R}\ \bibnamefont
  {McClean}}, \bibinfo {author} {\bibfnamefont {Masoud}\ \bibnamefont
  {Mohseni}}, \bibinfo {author} {\bibfnamefont {Patrick}\ \bibnamefont {van~der
  Smagt}}, \ and\ \bibinfo {author} {\bibfnamefont {Martin}\ \bibnamefont
  {Leib}},\ }\bibfield  {title} {\enquote {\bibinfo {title} {Layerwise learning
  for quantum neural networks},}\ }\href {https://arxiv.org/abs/2006.14904}
  {\bibfield  {journal} {\bibinfo  {journal} {arXiv preprint arXiv:2006.14904}\
  } (\bibinfo {year} {2020})}\BibitemShut {NoStop}%
\bibitem [{\citenamefont {Grant}\ \emph {et~al.}(2019)\citenamefont {Grant},
  \citenamefont {Wossnig}, \citenamefont {Ostaszewski},\ and\ \citenamefont
  {Benedetti}}]{grant2019initialization}%
  \BibitemOpen
  \bibfield  {author} {\bibinfo {author} {\bibfnamefont {Edward}\ \bibnamefont
  {Grant}}, \bibinfo {author} {\bibfnamefont {Leonard}\ \bibnamefont
  {Wossnig}}, \bibinfo {author} {\bibfnamefont {Mateusz}\ \bibnamefont
  {Ostaszewski}}, \ and\ \bibinfo {author} {\bibfnamefont {Marcello}\
  \bibnamefont {Benedetti}},\ }\bibfield  {title} {\enquote {\bibinfo {title}
  {An initialization strategy for addressing barren plateaus in parametrized
  quantum circuits},}\ }\href
  {https://quantum-journal.org/papers/q-2019-12-09-214/} {\bibfield  {journal}
  {\bibinfo  {journal} {Quantum}\ }\textbf {\bibinfo {volume} {3}},\ \bibinfo
  {pages} {214} (\bibinfo {year} {2019})}\BibitemShut {NoStop}%
\bibitem [{\citenamefont {Bharti}\ and\ \citenamefont
  {Haug}(2020{\natexlab{a}})}]{bharti2020iterative}%
  \BibitemOpen
  \bibfield  {author} {\bibinfo {author} {\bibfnamefont {Kishor}\ \bibnamefont
  {Bharti}}\ and\ \bibinfo {author} {\bibfnamefont {Tobias}\ \bibnamefont
  {Haug}},\ }\bibfield  {title} {\enquote {\bibinfo {title} {Iterative quantum
  assisted eigensolver},}\ }\href {https://arxiv.org/abs/2010.05638} {\bibfield
   {journal} {\bibinfo  {journal} {arXiv preprint arXiv:2010.05638}\ }
  (\bibinfo {year} {2020}{\natexlab{a}})}\BibitemShut {NoStop}%
\bibitem [{\citenamefont {Bharti}\ and\ \citenamefont
  {Haug}(2020{\natexlab{b}})}]{bharti2020quantumAS}%
  \BibitemOpen
  \bibfield  {author} {\bibinfo {author} {\bibfnamefont {Kishor}\ \bibnamefont
  {Bharti}}\ and\ \bibinfo {author} {\bibfnamefont {Tobias}\ \bibnamefont
  {Haug}},\ }\bibfield  {title} {\enquote {\bibinfo {title} {Quantum assisted
  simulator},}\ }\href {https://arxiv.org/abs/2011.06911} {\bibfield  {journal}
  {\bibinfo  {journal} {arXiv preprint arXiv:2011.06911}\ } (\bibinfo {year}
  {2020}{\natexlab{b}})}\BibitemShut {NoStop}%
\bibitem [{\citenamefont {Huembeli}\ and\ \citenamefont
  {Dauphin}(2021)}]{Huembeli2020Characterizing}%
  \BibitemOpen
  \bibfield  {author} {\bibinfo {author} {\bibfnamefont {Patrick}\ \bibnamefont
  {Huembeli}}\ and\ \bibinfo {author} {\bibfnamefont {Alexandre}\ \bibnamefont
  {Dauphin}},\ }\bibfield  {title} {\enquote {\bibinfo {title} {Characterizing
  the loss landscape of variational quantum circuits},}\ }\href
  {https://iopscience.iop.org/article/10.1088/2058-9565/abdbc9} {\bibfield
  {journal} {\bibinfo  {journal} {Quantum Science and Technology}\ } (\bibinfo
  {year} {2021})}\BibitemShut {NoStop}%
\bibitem [{\citenamefont {O’Brien}\ \emph {et~al.}(2019)\citenamefont
  {O’Brien}, \citenamefont {Senjean}, \citenamefont {Sagastizabal},
  \citenamefont {Bonet-Monroig}, \citenamefont {Dutkiewicz}, \citenamefont
  {Buda}, \citenamefont {DiCarlo},\ and\ \citenamefont
  {Visscher}}]{o2019calculating}%
  \BibitemOpen
  \bibfield  {author} {\bibinfo {author} {\bibfnamefont {Thomas~E}\
  \bibnamefont {O’Brien}}, \bibinfo {author} {\bibfnamefont {Bruno}\
  \bibnamefont {Senjean}}, \bibinfo {author} {\bibfnamefont {Ramiro}\
  \bibnamefont {Sagastizabal}}, \bibinfo {author} {\bibfnamefont {Xavier}\
  \bibnamefont {Bonet-Monroig}}, \bibinfo {author} {\bibfnamefont {Alicja}\
  \bibnamefont {Dutkiewicz}}, \bibinfo {author} {\bibfnamefont {Francesco}\
  \bibnamefont {Buda}}, \bibinfo {author} {\bibfnamefont {Leonardo}\
  \bibnamefont {DiCarlo}}, \ and\ \bibinfo {author} {\bibfnamefont {Lucas}\
  \bibnamefont {Visscher}},\ }\bibfield  {title} {\enquote {\bibinfo {title}
  {Calculating energy derivatives for quantum chemistry on a quantum
  computer},}\ }\href {https://www.nature.com/articles/s41534-019-0213-4#Equ10}
  {\bibfield  {journal} {\bibinfo  {journal} {npj Quantum Information}\
  }\textbf {\bibinfo {volume} {5}},\ \bibinfo {pages} {1--12} (\bibinfo {year}
  {2019})}\BibitemShut {NoStop}%
\bibitem [{\citenamefont {Gill}\ \emph {et~al.}(2019)\citenamefont {Gill},
  \citenamefont {Murray},\ and\ \citenamefont {Wright}}]{gill2019practical}%
  \BibitemOpen
  \bibfield  {author} {\bibinfo {author} {\bibfnamefont {Philip~E}\
  \bibnamefont {Gill}}, \bibinfo {author} {\bibfnamefont {Walter}\ \bibnamefont
  {Murray}}, \ and\ \bibinfo {author} {\bibfnamefont {Margaret~H}\ \bibnamefont
  {Wright}},\ }\href@noop {} {\emph {\bibinfo {title} {Practical
  optimization}}}\ (\bibinfo  {publisher} {SIAM},\ \bibinfo {year}
  {2019})\BibitemShut {NoStop}%
\bibitem [{\citenamefont {Mari}\ \emph {et~al.}(2020)\citenamefont {Mari},
  \citenamefont {Bromley},\ and\ \citenamefont
  {Killoran}}]{mari2020estimating}%
  \BibitemOpen
  \bibfield  {author} {\bibinfo {author} {\bibfnamefont {Andrea}\ \bibnamefont
  {Mari}}, \bibinfo {author} {\bibfnamefont {Thomas~R}\ \bibnamefont
  {Bromley}}, \ and\ \bibinfo {author} {\bibfnamefont {Nathan}\ \bibnamefont
  {Killoran}},\ }\bibfield  {title} {\enquote {\bibinfo {title} {Estimating the
  gradient and higher-order derivatives on quantum hardware},}\ }\href
  {https://arxiv.org/abs/2008.06517} {\bibfield  {journal} {\bibinfo  {journal}
  {arXiv preprint arXiv:2008.06517}\ } (\bibinfo {year} {2020})}\BibitemShut
  {NoStop}%
\bibitem [{\citenamefont {Mitarai}\ \emph {et~al.}(2018)\citenamefont
  {Mitarai}, \citenamefont {Negoro}, \citenamefont {Kitagawa},\ and\
  \citenamefont {Fujii}}]{mitarai2018quantum}%
  \BibitemOpen
  \bibfield  {author} {\bibinfo {author} {\bibfnamefont {K.}~\bibnamefont
  {Mitarai}}, \bibinfo {author} {\bibfnamefont {M.}~\bibnamefont {Negoro}},
  \bibinfo {author} {\bibfnamefont {M.}~\bibnamefont {Kitagawa}}, \ and\
  \bibinfo {author} {\bibfnamefont {K.}~\bibnamefont {Fujii}},\ }\bibfield
  {title} {\enquote {\bibinfo {title} {Quantum circuit learning},}\ }\href
  {\doibase 10.1103/PhysRevA.98.032309} {\bibfield  {journal} {\bibinfo
  {journal} {Phys. Rev. A}\ }\textbf {\bibinfo {volume} {98}},\ \bibinfo
  {pages} {032309} (\bibinfo {year} {2018})}\BibitemShut {NoStop}%
\bibitem [{\citenamefont {Schuld}\ \emph {et~al.}(2019)\citenamefont {Schuld},
  \citenamefont {Bergholm}, \citenamefont {Gogolin}, \citenamefont {Izaac},\
  and\ \citenamefont {Killoran}}]{schuld2019evaluating}%
  \BibitemOpen
  \bibfield  {author} {\bibinfo {author} {\bibfnamefont {Maria}\ \bibnamefont
  {Schuld}}, \bibinfo {author} {\bibfnamefont {Ville}\ \bibnamefont
  {Bergholm}}, \bibinfo {author} {\bibfnamefont {Christian}\ \bibnamefont
  {Gogolin}}, \bibinfo {author} {\bibfnamefont {Josh}\ \bibnamefont {Izaac}}, \
  and\ \bibinfo {author} {\bibfnamefont {Nathan}\ \bibnamefont {Killoran}},\
  }\bibfield  {title} {\enquote {\bibinfo {title} {Evaluating analytic
  gradients on quantum hardware},}\ }\href
  {https://journals.aps.org/pra/abstract/10.1103/PhysRevA.99.032331} {\bibfield
   {journal} {\bibinfo  {journal} {Physical Review A}\ }\textbf {\bibinfo
  {volume} {99}},\ \bibinfo {pages} {032331} (\bibinfo {year}
  {2019})}\BibitemShut {NoStop}%
\bibitem [{\citenamefont {Stewart}(2015)}]{stewart2015multivariable}%
  \BibitemOpen
  \bibfield  {author} {\bibinfo {author} {\bibfnamefont {James}\ \bibnamefont
  {Stewart}},\ }\href@noop {} {\emph {\bibinfo {title} {Multivariable
  calculus}}}\ (\bibinfo  {publisher} {Nelson Education},\ \bibinfo {year}
  {2015})\BibitemShut {NoStop}%
\bibitem [{\citenamefont {Mitarai}\ \emph {et~al.}(2020)\citenamefont
  {Mitarai}, \citenamefont {Nakagawa},\ and\ \citenamefont
  {Mizukami}}]{mitarai2020theory}%
  \BibitemOpen
  \bibfield  {author} {\bibinfo {author} {\bibfnamefont {Kosuke}\ \bibnamefont
  {Mitarai}}, \bibinfo {author} {\bibfnamefont {Yuya~O}\ \bibnamefont
  {Nakagawa}}, \ and\ \bibinfo {author} {\bibfnamefont {Wataru}\ \bibnamefont
  {Mizukami}},\ }\bibfield  {title} {\enquote {\bibinfo {title} {Theory of
  analytical energy derivatives for the variational quantum eigensolver},}\
  }\href
  {https://journals.aps.org/prresearch/abstract/10.1103/PhysRevResearch.2.013129}
  {\bibfield  {journal} {\bibinfo  {journal} {Physical Review Research}\
  }\textbf {\bibinfo {volume} {2}},\ \bibinfo {pages} {013129} (\bibinfo {year}
  {2020})}\BibitemShut {NoStop}%
\end{thebibliography}%

\begin{appendices}

\section{Explicit description of $d_{(\omega_l,N_k)}$}

In this appendix we first discuss how the parameter shift rule leads to the Pascal tree. Then, we provide analytical formulas for $d_{(\omega_l,N_k)}$.

Let us  consider the first and second order partial derivatives of the cost function with respect to the same angle. From the parameter shift rule of Eq.~\eqref{eq:pshift} we find
\small
\begin{align}
    \partial_i C(\thv)= \frac{1}{2}\Big[&\underbrace{C\left(\thv_{\overline{i}}, \theta_i^{(\frac{1}{2})}\right)}_{\times d_{(\frac{1}{2},1)}=1}-\underbrace{C\left(\thv_{\overline{i}}, \theta_i^{(-\frac{1}{2})}\right)}_{\times d_{(-\frac{1}{2},1)}=-1}\Big]\nonumber\label{eq:first}\\
        \partial_i^2 C(\thv)=\frac{1}{4}\Big[&\underbrace{C \left(\thv_{\overline{i}}, \theta_i^{(1)}\right)}_{\times d_{(1,2)}=1}+
   \underbrace{ C \left(\thv_{\overline{i}}, \theta_i^{(-1)}\right)}_{\times d_{(-1,2)}=1}-2\underbrace{C \left(\thv\right)}_{\times d_{(0,2)}=-2}    \Big]\nonumber\,.
\end{align}
\normalsize
where we can see that  $|d_{(0,2)}|=|d_{(-\frac{1}{2},1)}|+|d_{(\frac{1}{2},1)}|=2$. Similarly, if we were to take the third partial derivative with respect to $i$ we would find $|d_{(\pm1/2,3)}|=|d_{(0,2)}|+|d_{(\pm 1,2)}|$, and $|d_{(\pm3/2,3)}|=|d_{(\pm1,2)}|$. Note that this procedure forms the first four rows of the Pascal tree, which actually coincide with first four rows of the Pascal triangle. 
When taking the fourth partial derivative we have to take into account the fact that $C \left(\thv_{\overline{i}}, \theta_i^{(-2)}\right)=C \left(\thv_{\overline{i}}, \theta_i^{(2)}\right)=C \left(\thv\right)$,  since  $e^{-i\theta\sigma/2}$ is equal to $e^{-i(\theta+2\pi)\sigma/2}$ up to an  unobservable global phase. Hence, the fact that $\theta\equiv \theta^{(2)}(\text{mod } 2\pi)$ imposes a restriction on the width of the Pascal tree. Hence, following this procedure one can recover the entries in Fig.~\ref{fig:pascal}.

For arbitrary $\omega_l$ and $N_l$,  the coefficients $d_{\omega_l,N_k}$ can be analytically obtained as follows: If $N_k<2$ we have $d_{(\pm1,0)}=0$, $d_{0,0)}=1$, $d_{(\pm1/2,1)}=\pm 1$, and $d_{(\pm3/2,1)}=0$. 
Then, for  $N_k\geq 2$ 
\begin{align}\nonumber
    d_{(\omega_l,N_k)}=\begin{cases}
       (-1)^{\frac{N_k}{2}} 2^{N_k-1}\quad&\text{if}\quad \omega_l=0\,,\\
     \pm  (-1)^{\frac{N_k-1}{2}} 3\cdot 2^{N_k-3}\quad&\text{if}\quad \omega_l=\pm1/2\,,\\
    (-1)^{\frac{N_k-2}{2}}2^{N_k-2} \quad&\text{if}\quad \omega_l=\pm1\\
    \mp (-1)^{\frac{N_k-1}{2}} 2^{N_k-3}\quad&\text{if}\quad \omega_l= \pm 3/2\,,
    \end{cases}
\end{align}

Note that $\forall N_l$ we have $\sum_{\omega_l} |d_{(\omega_l,N_k)}| =2^{N_k}$.
\end{appendices}

\end{document}